\documentclass[11pt,onecolumn]{article}

\usepackage{hyperref}       
\usepackage{url}            
\usepackage{booktabs}       
\usepackage{amsfonts}       
\usepackage{nicefrac}       
\usepackage{microtype}      
\usepackage{xcolor}         

\usepackage{relsize,amssymb}
\usepackage{stmaryrd}
\usepackage{ragged2e}

\usepackage{amsthm}
\usepackage{amsmath}
\usepackage{soul}
\usepackage{graphicx}
\usepackage{array}
\usepackage{subcaption}
\usepackage[vlined,linesnumbered,ruled,resetcount]{algorithm2e}

\usepackage[left=1.0in,top=0.5in,right=1.05in,bottom=1.0in,nohead,textheight=10in,footskip=0.3in]{geometry}

\newtheorem{theorem}{Theorem}
\newtheorem{corollary}[theorem]{Corollary}

\newtheorem{definition}[theorem]{Definition}
\newtheorem{proposition}[theorem]{Proposition}

\def\calY{\mathcal{Y}}
\def\calI{\mathcal{I}}
\def\calL{\mathcal{L}}
\def\calF{\mathcal{F}}
\def\calX{\mathcal{X}}
\def\calG{\mathcal{G}}
\def\calV{\mathcal{V}}
\def\calE{\mathcal{E}}

\def\calA{\mathcal{A}}
\def\calC{\mathcal{C}}
\def\thetabar{{\overline{\theta}}}
\DeclareMathOperator*{\argmin}{argmin}

\newcommand{\eqdef}{{\stackrel{\mbox{\tiny \tt ~def~}}{=}}}

\def\myparagraph#1{\vspace{0pt}\noindent{\bf #1~~}}

\let\citep\cite

\begin{document}


\title{\Large\bf  }
\author{}
\title{
\Large\bf  \vspace{-20pt} Tighter relaxations for MAP-MRF optimization \\ via Singleton Arc Consistency
}
\author{
  Asaf Lev-Ran, Pavel Arkhipov, Vladimir Kolmogorov\\
  \\
  Institute of Science and Technology Austria (ISTA)\\
  \texttt{\{asaf.lev-ran,pavel.arkhipov,vnk\}@ist.ac.at}
}
\date{}
\maketitle
\begin{abstract}
    We consider the MAP-MRF inference task, that is, minimizing a function of discrete variables represented as a sum of unary and pairwise terms. 
    A prominent approach for tackling this NP-hard problem in practice is to solve its natural LP relaxation and then iteratively tighten the relaxation
    by adding {\em clusters}. Based on some theoretical observations, we propose a new technique for identifying such clusters. It works by running 
    the {\em Singleton Arc Consistency} algorithm in a certain CSP instance. Experimental results indicate that the new tightening technique outperforms
    the previous approach by~\cite{sontag_frustrated_cycles} that searches for {\em frustrated cycles}.
    Our code will be made available at \url{https://github.com/vnk-ist/MAP-MRF/}.
\end{abstract}

\section{Introduction}
This paper focuses on the problem of minimizing functions of discrete variables represented as a sum of unary and pairwise terms:
\begin{equation}
    \label{eq_objective_original}
    \min_{x\in\calX} f[\theta](x),\qquad\quad f[\theta](x)= \sum_{v \in V} \theta_v(x_v) + \sum_{\{u, v\} \in E} \theta_{uv}(x_u, x_v) 
\end{equation}
The problem is specified by a graph $G = (V, E)$, a discrete set of labels $\calX_v$ for each node $v\in V$,
unary costs $\theta_v:\calX_v\rightarrow\mathbb R$ for each $v\in V$, and pairwise costs $\theta_{uv}:\calX_u\times\calX_v\rightarrow\mathbb R$
for each $(u,v)\in E$.
The set of labelings is defined as $\calX=\prod_{v\in V} \calX_v$.

Under a probabilistic interpretation, this problem is known as
{\em MAP inference in MRFs or CRFs}, i.e., maximum a posteriori inference in
Markov random fields or conditional random fields, in the computer vision and
machine learning communities~\citep{Blake2011MRFVision}. In the constraint
programming and constraint satisfaction communities, problem~\eqref{eq_objective_original} is known as {\em Weighted CSPs} or
{\em Cost Function Networks}~\citep{Cooper2010SoftArcConsistency}.

MAP-MRF optimization was a prominent approach in computer vision before the
deep-learning era, especially for problems that could be formulated as discrete
labelings with spatial smoothness, such as image segmentation, stereo, restoration,
and object labeling~\citep{Szeliski2008ComparativeMRF,Blake2011MRFVision}.
More recently, MRF/CRF ideas are often used as structured components coupled to
learned representations rather than as standalone hand-engineered models. Examples
include learned potentials and neural message passing for stereo matching
\citep{Guan2024NeuralMRF}, learned MRF layers over discrete image tokens for
efficient text-to-image generation~\citep{Jayasumana2024MarkovGen},
SAM-based area matching formulated as an MRF energy minimization problem solved
by graph cuts~\citep{Zhang2024MESA}, neuralized MRFs for interaction-aware human
trajectory prediction~\citep{Fang2025NeuralizedMRF}, and voxel-adaptive message
passing for deformable CT registration~\citep{Zhang2025VoxelOpt}.

Applications outside computer vision include sequence labeling tasks such as
named entity recognition, part-of-speech tagging, and chunking
\citep{MaHovy2016BiLSTMCNNCRF,Lample2016NeuralNER,Wei2021MaskedCRF}, protein
design~\citep{Traore2013CPD,Allouche2014CPDOptimization}, and recent deep
protein-design methods that incorporate MRF-style components
\citep{Ren2024CarbonDesign,Ren2024CarbonNovo}. Related pairwise MAP formulations
also arise in Markov-logic knowledge-base reasoning~\citep{Fang2023MLN4KB},
while the binary case is widely studied in operations research as QUBO or
UBQP~\citep{Kochenberger2014UBQPSurvey,Glover2022QUBO}.


\myparagraph{LP relaxations and tightening} Although problem~\eqref{eq_objective_original} is NP-hard in general,
it has been demonstrated that many practical instances can be solved to optimality.
A prominent approach is {\em cluster-based tightening}~\citep{sontag_tightening}
which starts by solving an LP relaxation of~\eqref{eq_objective_original} and then iteratively tightens
the relaxation by adding clusters of nodes. The success of this approach depends crucially on which
clusters are added. Early approaches chose clusters from a fixed dictionary, e.g.\ the set of all triplets
(3-cycles in the graph)~\citep{sontag_tightening,komodakis_repairing_cycles,batra_LPDG}.
\cite{sontag_frustrated_cycles} showed how to avoid fixed dictionaries by
proposing an algorithm for computing ``frustrated cycles'', while showing that after adding these cycles the dual objective can be improved.

The main contribution of this paper is a new algorithm for computing clusters. 
It is based on running a {\em Singleton Arc Consistency} (SAC) algorithm on a certain CSP instance.
We prove several properties of this algorithm. Roughly speaking, they can be formulated
as follows: (1) If SAC produces a cluster then after adding this cluster either the dual objective can be improved or some fractional vertices are cut off from the local polytope.
(2) If the instance has at least one frustrated cycle then SAC will find a cluster.
For more formal statements we refer to Sections~\ref{sec_theory} and~\ref{sec:frustrated}.

Our experimental results show our SAC tightening procedure outperforms the method in~\citep{sontag_frustrated_cycles}
by consistently finding the same or better bounds.

\myparagraph{Other related work}
\cite{sontag_outer_bounds} proposed to tighten relaxations of MAP-MRF problems
by adding cutting planes that generalize  {\em odd-cycle inequalities} for MAX-CUT.
Their separation algorithm works in the primal domain,
and requires a (fractional) primal solution of the relaxation. In contrast,
we focus on tightening techniques that use only dual solutions (``reparameterization''),
since relaxations of large-scale problems are often solved via message-passing techniques
that work only in the dual domain.

Tightening techniques require introducing high-order terms, which can become expensive when the number of labels
is large. \cite{sontag_coarse_partitions} address this issue by proposing to use {\em coarse partitions} on labels
for high-order terms. Note that this approach is orthogonal to the question studied in this paper.

\cite{werner_reparameterization_is_duality} presented an approach
based on the language of Constraint Satisfaction Problems (CSPs), by showing 
that all unsatisfiable CSP subproblems could potentially tighten the relaxation.
However, no concrete algorithm for identifying such subproblems was proposed.

In this work we use a particular CSP algorithm, namely Singleton Arc Consistency (SAC).
This algorithm was recently used in \citep{dlask_superreparameterization}
for identifying {\em superreparameterizations} that improve the lower bound.
At each iteration they replace current vector $\theta$ with another vector $\thetabar$
that satisfies $f[\thetabar](x)\le f[\theta](x)$ for all labelings $x$,
and guarantees that $\thetabar$ gives a better lower bound than $\theta$.
This approach is not comparable with the tightening techniques that we consider:
it can be computationally less costly since it does not introduce any high-order terms,
but it may also decrease the minimum  $\min_{x\in\calX} f[\theta](x)$
(which would never be ``undone'' in later iterations).

Note that a simpler Arc Consistency (AC) algorithm has been used for approximately solving the standard LP relaxation of~\eqref{eq_objective_original} (without tightening). In particular, an AC-based algorithm for establishing the {\em Virtual Arc Consistency} has been given in~\citep{CooperGSSZ08,Cooper2010SoftArcConsistency}, and is similar to the {\em Augmenting DAG} algorithm~\citep{Schlesinger1976,werner_reparameterization_is_duality}.

There are also alternative tightening approaches based on more expensive computational tools such as semidefinite programming relaxations, see e.g.~\citep{pmlr-v162-durante22a}.
We limit the scope of this work to LP relaxations, which arguably can scale to larger problems.

A survey of tightening techniques for MAP-MRF optimization can be found in~\citep{kannan2019tighter}.

\section{Background}\label{sec:background}
The cluster-based tightening approach~\citep{sontag_tightening} starts with the objective function in~\eqref{eq_objective_original}
and then repeats the following procedure:
\begin{itemize}
\item Add some zero-valued higher-order terms to the current objective, leading to the following optimization problem:
\begin{equation}
    \label{eq:f} 
    \min_{x\in\calX} f[\theta](x),\qquad\quad f[\theta](x)= \sum_{\alpha\in\calF} \theta_\alpha(x_\alpha)  
\end{equation}
Here $\calF\subset 2^V$ is a set of non-empty subsets of $V$, called {\em factors}.
For a factor $\alpha \in \calF$ we denote $\calX_\alpha\eqdef\prod\limits_{v\in\alpha}\calX_v$,
and for a labeling $x\in\calX$ we let $x_\alpha\in \calX_\alpha$ 
be the restriction of $x$ to $\alpha$.
By construction, $\theta_\alpha(x_\alpha)=0$ for all $\alpha\in\calF,x_\alpha\in\calX_\alpha$ with $|\alpha|\ge 3$.
\item Approximately solve a natural LP relaxation of~\eqref{eq:f}, or rather its dual.
\end{itemize}
Note that although discrete optimization problems~\eqref{eq_objective_original} and~\eqref{eq:f} are equivalent,
their LP relaxations are not: the lower bound given by the relaxation of~\eqref{eq:f} is the same or larger than that of~\eqref{eq_objective_original}.
Below we describe these two steps in more detail. 

\subsection{LP relaxation}
A class of relaxations of~\eqref{eq:f} has been presented in~\citep{werner_reparameterization_is_duality}.
Each relaxation is specified by a directed acyclic graph $(\calF,J)$
where $J$ is a set of pairs of the form $(\alpha,\beta)$ with $\alpha,\beta\in\calF$ and $\beta\subset\alpha$.
Define the {\em $J$-based local polytope} $\calL(J)$ as follows:
\begin{equation}
\hspace{1pt}\calL(J)=\left\{ \mu\ge 0 \;
\begin{picture}(1,0)
  \put(3,-20){\line(0,1){44}}
\end{picture}
\; \mbox{\begin{tabular}{ll} 
\;\;\; $\;\mathlarger\sum\limits_{x_\alpha}\hspace{14pt}\mu_\alpha(x_\alpha)=1$ & $\forall \alpha\in\calF\!\!\!$ \\
$\mathlarger\sum\limits_{x_{\alpha}:x_\alpha\sim x_\beta}\mu_\alpha(x_\alpha) =\mu_\beta(x_\beta)$  \hspace{20pt} & $\forall(\alpha,\beta)\in J,x_\beta\!\!\!$ 
\end{tabular}} \right\}\hspace{-5pt}
\label{eq:LocalPolytope}
\end{equation}
Here notation $x_\alpha\sim x_\beta$ for factors $\alpha,\beta\subseteq V$
means that $x_\alpha\in \calX_\alpha$ and $x_\beta\in \calX_\beta$
are labelings that agree on the overlap $\gamma=\alpha\cap\beta$, i.e. $(x_\alpha)|_\gamma=(x_\beta)|_\gamma$.
A relaxation of~\eqref{eq:f} is now given by
\begin{equation}\label{eq:LP}
\min_{\mu\in\calL(J)} \sum_{\alpha\in\calF,x_\alpha\in\calX_\alpha} \theta_\alpha(x_\alpha)\mu(x_\alpha)
\end{equation}

In this work we will tighten the relaxation only with triplets, i.e.\ only add subsets $\alpha\subseteq V$ with $|\alpha|=3$.
If such $\alpha$ contains a pair of nodes $u,v$ with $\{u,v\}\notin E$ then we also add such factor $\{u,v\}$ to $\calF$
(with zero costs). Set $J$ will always be the Hasse diagram of the poset on $\calF$: 
$J=\{(\alpha,\beta)\::\alpha,\beta\in \calF,\beta \subset\alpha, |\alpha|=|\beta|+1\}$.
Note that when adding a new factor we check whether it already exists; if yes, then we do not duplicate it.


\myparagraph{Reparameterizations and the dual problem}
It is straightforward to check that the following transformation does not change the function $f[\theta]$:
pick edge $(\alpha,\beta)\in J$, labeling $x_\beta\in\calX_\beta$, value $\delta\in\mathbb R$,
and then update $\theta_\beta(x_\beta)\rightarrow \theta_\beta(x_\beta)+\delta$ and $\theta_\alpha(x_\alpha)\rightarrow \theta_\alpha(x_\alpha)-\delta$
for all $x_\alpha\sim x_\beta$. Such operation is called a {\em reparameterization}.
The set of all reparameterizations can be described by a vector $\lambda$
with components $\lambda_{\alpha\beta}(x_\beta)$ for all $(\alpha,\beta)\in J$, $x_\beta\in\calX_\beta$ using the following rule:
\begin{equation}\label{eq:rep}
\thetabar_\beta(x_\beta)=\theta_\beta(x_\beta)+\sum_{(\alpha,\beta)\in J} \lambda_{\alpha\beta}(x_\beta)- \sum_{(\beta,\gamma)\in J} \lambda_{\beta\gamma}(x_\gamma)
\end{equation}
where $x_\gamma$ in the last term is the restriction of $x_\beta$ to $\gamma$.
Vector $\lambda_{\alpha\beta}\in\mathbb R^{\calX_\beta}$ is also known as a {\em message} from $\alpha$ to $\beta$.
We will write $\thetabar=\theta^\lambda$ for vector $\thetabar$ obtained from $\theta,\lambda$ as in~\eqref{eq:rep}.
Clearly, we have $f[\theta^\lambda](x)=f[\theta](x)$ for any $\theta,\lambda$ and $x\in\calX$.

For each vector $\theta$ the following expression gives a lower bound on $\min_{x\in\calX} f[\theta](x)$:
\begin{equation}
\Phi(\theta)=\sum_{\alpha\in\calF} \min_{x_\alpha} \theta_\alpha(x_\alpha)
\end{equation}
The dual problem to~\eqref{eq:LP} is equivalent to 
maximizing this bound over messages $\lambda$ (see~\citep{werner_reparameterization_is_duality}):
\begin{equation}\label{eq:dual}
\max_\lambda \Phi(\theta^\lambda)
\end{equation}

\myparagraph{Message passing algorithms} Solving~\eqref{eq:LP} or its dual~\eqref{eq:dual} exactly
using generic LP solvers can be too expensive in practice for large-scale problems.
A large number of approximate techniques have been proposed for this problem.
A prominent family are {\em message passing} algorithms which maximize the bound $\Phi(\theta^\lambda)$
via a block-coordinate ascent~\citep{trws,mplp,werner_reparameterization_is_duality,sontag_tightening,srmp,Swoboda2017DualAscent}.
In this work we use the SRMP algorithm~\citep{srmp}.
Note that for pairwise MAP-MRF problems it is equivalent to TRW-S~\citep{trws}.

\subsection{Cluster pursuit via frustrated cycles}\label{sec:FR:intro}
Next, we discuss which clusters to add. As described in the introduction,
the technique which is most relevant to our paper is~\citep{sontag_frustrated_cycles}.
It works as follows. First, it selects a set $\Pi_v$ for each node $v\in V$,
where each $\pi=(\pi^+,\pi^-)\in\Pi_v$ is a partition of $\calX_v$ into two non-empty sets: $\calX_v=\pi^+\sqcup \pi^-$.
Let $\Pi$ be the union of partitions $\Pi_v$ over $v\in V$ (which we assume to be disjoint).
Now consider partitions $\pi_u\in \Pi_u$, $\pi_v\in \Pi_v$.
If $\{u,v\}\in E$ then define 
$$
w(\pi_u,\pi_v)=
 \min_{\substack{(a,b)\in\calX_u\times \calX_v \\ \llbracket a\in\pi_u^+\rrbracket\:\ne\: \llbracket b\in\pi_v^+\rrbracket}}\theta_{uv}(a,b)
-\min_{\substack{(a,b)\in\calX_u\times \calX_v \\ \llbracket a\in\pi_u^+\rrbracket\:=\: \llbracket b\in\pi_v^+\rrbracket}}\theta_{uv}(a,b)
$$
where $\theta$ is the current reparameterization, and $\llbracket \cdot\rrbracket$ is the Iverson bracket notation. Otherwise set $w(\pi_u,\pi_v)=0$.
Let $\calG=(\Pi,\calE)$ be an undirected graph with the edges
 $\calE=\{\{\pi_u,\pi_v\}\::\:|w(\pi_u,\pi_v)|>\varepsilon\}$ where $\varepsilon$ is a fixed positive constant.
 Call edge $\{\pi_u,\pi_v\}\in\calE$ {\em positive} if $w(\pi_u,\pi_v)>\varepsilon$, and {\em negative} if 
 $w(\pi_u,\pi_v)<-\varepsilon$. Now suppose that $\calG$ has a cycle $\calC$ that contains an odd number of negative edges;
 such a cycle is called {\em frustrated}. Let $C$ be the cycle in $G$ obtained from $\calC$ by replacing each node $\pi_v\in\Pi_v$ of the cycle with $v$.
 It can be seen that for any labeling $x\in\calX$ there exists an edge $\{u,v\}\in C$ 
 such that $\theta_{uv}(x_u,x_v)>\min\limits_{(a,b)\in\calX_u\times\calX_v}\theta_{uv}(a,b)+\varepsilon$. Motivated by this observation, \cite{sontag_frustrated_cycles}
 propose to tighten the relaxation using  cycle $ C$. This can be achieved by first triangulating $C$ and then
 adding each triplet as a high-order term. As \cite{sontag_frustrated_cycles} observe,
  after this operation there exists a reparameterization that increases the lower bound by at least $\varepsilon$.

\cite{sontag_frustrated_cycles} also show that frustrated cycles in $\calG$ can be found in linear time.
This defines an algorithm for tightening relaxations.
We discuss further details of this algorithm in Section~\ref{sec:frustrated}.

\subsection{CSPs and Singleton Arc Consistency}
A {\em Constraint Satisfaction Problem} instance is given by a set of nodes $V$, a discrete domain $\calX=\prod_{v\in V}\calX_v$,
a set of factors $\calF\subset 2^V$, and relations $\calY_\alpha\subseteq \calX_\alpha$ for each $\alpha\in\calF$.
We will denote the instance  as a tuple $\calI=(V,\calX,\calF,\calY)$, or just as $\calY$ when $V,\calX,\calF$ are clear from the context.
We will write $\calY\preceq \calY'$ if $\calY_\alpha\subseteq \calY'_\alpha$ for all $\alpha\in\calF$.
A {\em solution} to $(V,\calX,\calF,\calY)$ is a labeling $x\in\calX$ satisfying $x_\alpha\in \calY_\alpha$ for all $\alpha\in\calF$.
The instance is called {\em non-empty} if $\calY_\alpha\ne\varnothing$ for all $\alpha\in\calF$, and {\em satisfiable}
if it has at least one solution.

\begin{definition}
Consider CSP instance $\calI=(V,\calX,\calF,\calY)$ and graph $(\calF,J)$.
Instance $\calI$ is called {\em $J$-arc-consistent} if every $(\alpha,\beta)\in J$
satisfies $\{x_\beta\in\calX_\beta\::\:x_\alpha\sim x_\beta{\mbox ~for~\!~some~ }x_\alpha\in \calY_\alpha\}=\calY_\beta$.
It is {\em $J$-consistent} if there exists non-empty instance $\calY'\preceq \calY$ which is $J$-arc-consistent.
\end{definition}

For a vector $\theta$ and value $\varepsilon\ge 0$ we define the CSP instance $CSP_\varepsilon(\theta)=(V,\calX,\calF,\calY)$ by setting 
$\calY_\alpha=\{x_\alpha\::\:\theta_\alpha(x_\alpha) \le \min_{y_\alpha\in\calX_\alpha}\theta_\alpha(y_\alpha)+\varepsilon\}$.
The following fact is well-known~\citep{werner_reparameterization_is_duality,Cooper2010SoftArcConsistency,srmp}.
\begin{proposition} Consider instance $\calI=CSP_0(\theta)$ for a vector $\theta$. \\
{\rm (a)} If $\calI$ is satisfiable then $\Phi(\theta)=\min_{x\in\calX} f[\theta](x)$,
and $x\in\calX$ is a solution of $\calI$ if and only if $x\in\argmin_{x\in\calX} f[\theta](x)$. \\
{\rm (b)} Reparameterization $\calI$ is $J$-consistent if and only if $\theta$ is a local maximum of a block-coordinate
ascent algorithm such as SRMP, i.e.\ no sequence of operations of SRMP can increase $\Phi(\theta)$.
\end{proposition}

\subsection{2-CSPs}\label{sec:twoCSP}
We say that instance $\calI=(V,\calX,\calF,\calY)$ is a {\em $k$-CSP} if $|\alpha|\le k$ for all $\alpha\in\calF$.
We will be mainly working with the 2-CSP instance $CSP_\varepsilon^{\le 2}(\theta)$
obtained from $CSP_\varepsilon(\theta)$ by removing all factors $\alpha$ with $|\alpha|\ge 3$.
In this section we discuss some algorithms for 2-CSPs.
We will assume that $\calF=V\cup E$ where $E\subseteq\binom{V}2$ is a set of edges.
We say that it is {\em arc-consistent} if it is $J$-arc-consistent
for the set of edges $J=\{ev\::\:e\in E,v\in e\}$.

\myparagraph{Arc consistency (AC) algorithm}
Let us consider the following problem: given instance $\calI=(V,\calX,V\cup E,\calY)$,
compute a maximal instance $\calY'\subseteq \calY$ which is arc-consistent. Clearly, this can be done by
repeating the following two operations while possible.
\begin{enumerate}
\item Pick pair $(a,b)\in \calY_{uv}$ for $\{u,v\}\in E$ s.t.\ $a\notin \calY_u$.
Remove $(a,b)$ from $\calY_{uv}$.
\item Pick label $a\in \calY_u$ for $\{u,v\}\in E$ s.t.\ $(a,b)\notin \calY_{uv}$ for all $b\in\calX_v$.
Remove $a$ from~$\calY_u$.
\end{enumerate}

These operations can be implemented in $O(|E|\cdot(\max_v|\calX_v|)^2)$ time~\citep{bessiere2005optimal}.
Furthermore, if it produces an empty CSP (i.e.\ if the domain of one of the variables becomes empty) then
it is possible to find a {\em minimal} sequence of such operations in the same time.

\myparagraph{Singleton arc consistency (SAC)}
Instance $\calI=(V,\calX,V\cup E,\calY)$ is called {\em singleton arc consistent}
if the following holds for each $v\in V$ and $a\in \calY_v$:
fixing the label of $v$ to $a$ (i.e.\ replacing $\calY_v$ with $\{a\}$) and then running the arc consistency algorithm gives a non-empty instance~\citep{debruyne1997some}.
Such consistency can be achieved by repeating the following operation: pick node $v\in V$ and label $a\in \calY_v$,
fix the label of $v$ to $a$, run the AC algorithm; if it results in an empty CSP then remove $a$ from $\calY_v$. Efficient implementations of such an algorithm have been studied in~\citep{bartak2004new}.

\section{SAC tightening theory}
\label{sec_theory}
Let $\calF$ be the current set of factors, and $\theta$ be the current reparameterization. 
Suppose that fixing node $r\in V$ to label $s\in\calY_r$ in the CSP instance $\calI=(V,\calX,\calF,\calY)=CSP_{\varepsilon}^{\le 2}(\theta)$ and then running the arc consistency (AC) algorithm makes the CSP instance empty. From the results of~\cite{werner_reparameterization_is_duality} we can conclude that after adding factor $\alpha$ to the relaxation the lower bound can be strictly increased, where
$\alpha$ is the set of all nodes that were touched by the AC algorithm. In this section we will refine this result. Based on this refinement, in the next section we will then propose a specific algorithm for tightening. The proofs are given in the appendix.

Throughout this section, let $\{a_i\}_{i=1,\ldots,K}$ be the sequence of labels that were removed during AC, with $a_i \in \calX_{v_i}$. Let $u_i$ be the vertex that caused the deletion of $a_i$, i.e. $\{u_i, v_i\} \in E$, and for no $b \in R_{u_i}$ we had $(b, a_i) \in R_{u_i v_i}$. The value of $\varepsilon$ is fixed throughout the section. The AC run can terminate either with the deletion of $s$ (in this case, $a_K = s$) or with the deletion of the last label in some other vertex $v \neq r$. In the latter case, we artificially extend the sequence $a_i$ with $s$. Thus, we now always have $a_K = s$, $v_K = r$. Consider the relaxation obtained from the current one by adding triplets $\{r,u_i,v_i\}$ for $i=1,\ldots,K$ to $\calF$,
as described in Section~\ref{sec:background}.
Assume that vector $\theta$ is set to zero for newly added factors, and let $\calF'$ be the new set of factors. 
We can prove that our tightening is indeed a strict tightening, in the sense of the theorem stated below.
\begin{theorem}
    \label{th_strict_tightening}
    There exists a reparameterization $\thetabar$ of $\theta$ such that $\min_{x_\alpha}\thetabar_\alpha(x_\alpha)\ge \min_{x_\alpha}\theta_\alpha(x_\alpha)$ for all $\alpha\in\calF'$, and $\thetabar_r(s)>\theta_r(s)$.
\end{theorem}



In the above theorem, we saw that there exists a reparameterization that increases $\thetabar_r(s)$ by a sufficiently small positive number. The construction that we used in the proof might be wasteful. Let us discuss by how much we can improve $\thetabar_r(s)$ in the general case, if we use the same reparameterization defined by the AC run, but fine-tune the amounts for the reparameterization.

We need to define an auxiliary directed graph $\calG = (\calV, \calE)$. The vertex set of this graph is the set of all deleted labels: $\calV = \{a_i\}_{i=1}^K$. The directed edges consist of labels $(a_i, a_j)$ for $j > i$ if $a_i \in \calX_{u_j}$. Notice that $\calG$ is a directed acyclic graph, since the ordering of the labels $\{a_i\}$ in the chronological order of their deletion is also a topological ordering in $\calG$. 

Let us define another auxiliary directed graph $\calG' = (\calV', \calE')$ by applying a sequence of transformations to $\calG$. Every node of $\calG'$ is going to be a subset of labels of some vertex of the original graph $G$. Start with $\calG' = \calG$. In the beginning, all the nodes of $\calG'$ are single-element sets of labels. While possible, apply the following operation: contract two nodes $A, B$ into a single node $A \cup B$, if the following three properties hold. (1): $A$ and $B$ are subsets of the same label set $\calX_v$. (2): After the contraction, for every incoming edge $(C, A \cup B)$, we have that $C$ is a subset of the same label set $\calX_u$. (3): $\calG'$ remains acyclic.
Notice that the resulting $\calG'$ may not be unique, as it may depend on the order of the contractions. Since $\calG$ had a unique sink $a_K = s$, also $\calG'$ has a unique sink $S = \{s\}$.

Let us recursively define a function $B : \calV' \to \mathbb{N}$ called the \emph{branching factor}. $B(S) := 1$, and for any other node $A$, set $B(A) = \sum\limits_{C : (A, C) \in \calE'} B(C)$. Since $\calG'$ is acyclic by construction, this is indeed a correct recursive definition. Finally, we can state a competitive bound on how much one can increase $\thetabar_r(s)$ in the reparameterization.

\begin{theorem}
    \label{th_tightening_group_branching}
    There exists a reparameterization $\thetabar$ of $\theta$ such that $\min_{x_\alpha}\thetabar_\alpha(x_\alpha)\ge \min_{x_\alpha}\theta_\alpha(x_\alpha)$ for all $\alpha\in\calF'$, and $\thetabar_r(s) \geq \theta_r(s) + \frac{\varepsilon}{\max\limits_{A \in \calV'} B(A)}$.
\end{theorem}

As a corollary, we can conclude that in the practically common case when the AC deletes labels around a cycle, we can increase $\thetabar_r(s)$ by $\varepsilon$.
\begin{corollary}
    \label{cor_cycle}
    Let $w_1, \ldots, w_d$ be a cycle in $G$, where $w_1 = r$. If for all $i$, we have $(u_i, v_i) = (w_j, w_{j+1})$ (i.e. the AC deletes labels along this cycle), then there exists a reparameterization $\thetabar$ of $\theta$ such that $\min_{x_\alpha}\thetabar_\alpha(x_\alpha)\ge \min_{x_\alpha}\theta_\alpha(x_\alpha)$ for all $\alpha\in\calF'$, and $\thetabar_r(s) \geq \theta_r(s) + \varepsilon$.
\end{corollary}

\section{SAC tightening algorithm}
\label{sec_algorithm_description}
Results in the previous section suggest a natural algorithm for tightening. First, we will describe a naive version, and then discuss practical improvements.
Below ``AC'' means the arc consistency algorithm described in Section~\ref{sec:twoCSP}.

\begin{algorithm}[H]
  \DontPrintSemicolon
\SetNoFillComment
choose parameter $\varepsilon>0$, define CSP instance
$\calI=(V,\calX,\calF,\calY)=CSP_\varepsilon^{\le 2}(\theta)$. 
\\
update $\calI$ by running AC.
\\
\For{$r\in V$, $s\in \calY_r$}
{
fix node $r$ to $s$, run AC in $\calI$. If this causes inconsistency 
then find a minimal set of operations that cause inconsistency and construct a set of triplets $\calA_{rs}=\{\alpha_1,\ldots,\alpha_k\}$ as described in Section~\ref{sec_theory}; otherwise set $\calA_{rs}=\varnothing$.
Restore instance $\calI$.
}
for $r\in V$ define $\calL_r=\{s\in \calY_r\::\calA_{rs}\ne\varnothing\}$.
Define total order $\preceq$ on $V$ as follows: 
(1) if $\calL_r=\calY_r$ and $\calL_{r'}\ne \calY_{r'}$
then $r\prec r'$;
(2) otherwise if $|\calA_r|<|\calA_{r'}|$ then $r\prec r'$.
\\ set $\calA\leftarrow\varnothing$, then go through $r\in V$ in the order $\preceq$ and add $\calA_r$ to $\calA$ if $\calA\cap\calA_r=\varnothing$.~\footnote{
Here we actually assume that $\calA_r$  also contains ``triplets'' of the form $\{r,r,v\}$, if edge $\{r,v\}$ was touched during AC.
When defining the order $\preceq$, we ignore such ``triplets'' when counting $|\calA_r|$.
}
 \\
return $\calA$
      \caption{{\sc FindTriplets}.
      }\label{alg:SAC}
\end{algorithm}
By results in the previous section,
after adding $\calA_r=\bigcup_{s\in \calY_r}\calA_{rs}$
there exists a reparameterization $\thetabar$ such that $\min_{x_\alpha}\thetabar_\alpha(x_\alpha)\ge \min_{x_\alpha}\theta_\alpha(x_\alpha)$
for all $\alpha\in\calF$ and $\thetabar_r(a) > \theta_r(a)$ for all $a\in\calL_r$.
In particular, if $\calL_r=\calY_r$ then
 $\thetabar$ gives a better lower bound:  $\Phi(\thetabar)>\Phi(\theta)$.
This observation motivates the sorting rule (1) in line 5: it prioritizes sets of triplets which are guaranteed to increase the bound.
Note, we prune some triplets to avoid a scenario when we add triplets supported by the same cycle.

Next, we discuss some implementational choices.

\myparagraph{AC algorithm}
We used the AC3 algorithm~\citep{Mackworth1977} whose worst-case complexity is $O(|E|\cdot(\max_v|\calX_v|)^3)$. Although this is worse than the $O(|E|\cdot(\max_v|\calX_v|)^2)$ complexity of the AC2001/3.1 algorithm~\citep{bessiere2005optimal}, AC3 is simpler and may have an advantage when the number of labels is small. It maintains a queue of directed edges $(u,v)$, and processes them in certain order. To simulate a BFS-like search, we use the FIFO order for this queue. In our experiments, this tends to produce smaller inconsistent subproblems.
Importantly, we limit the search to depth $d_{\max}=3$ to reduce the complexity (when running AC in line 4).
This means that the algorithm is able to find frustrated cycles of length $2d_{\max}=6$.


\myparagraph{Choice of $\varepsilon$}
We maintain the current value of $\varepsilon$ (initially $\varepsilon=0.1$). At each tightening stage we run Algorithm~\ref{alg:SAC} for values $\varepsilon_i=2^{-i}\varepsilon$ ($i=0,1,2,\ldots$) and compute corresponding sets $\calA^{(i)}$, where for the $i$-th run set $\calA$ at line 6 is initialized via $\calA\leftarrow\calA^{(i-1)}$.
For each $i$ we then try to add computed triplets to $\calF$, and compute the number $k_i$ of those triplets that were not yet  present in $\calF$. If $k_i< 2 k_{i-1}$ then we stop and update $\varepsilon=\varepsilon_{i-1}$.
If $\varepsilon$ becomes smaller than $10^{-6}$ then we increase $d_{\max}$ by 1 and reset $\varepsilon=0.1$.

\section{Frustrated cycles algorithm}\label{sec:frustrated}
In this section we specify details of our implementation of the algorithm of~\cite{sontag_frustrated_cycles} described in Section~\ref{sec:FR:intro}. 
The first question is how to set partitions $\Pi_v$ for $v\in V$.
To match the SAC tightening algorithm, we compute $\calY_v=\{a\in\calX_v\::\:\theta_v(a)\le\min_b \theta_v(b)+\varepsilon\}$ and then for each $a\in\calY_v$ add the partition $(\{a\},\calX_v-\{a\})$ to $\Pi_v$.

Next, we need to compute frustrated cycles, i.e.\ cycles in the signed graph $\calG$ which have an odd number of negative edges.
To do this, \cite{sontag_frustrated_cycles} first pick a spanning rooted tree $T$ in $\calG$ (namely, a BFS tree starting from a random node).
They observe that for each edge $(u,v)$ in $\calG$ one can check in $O(1)$ time whether the fundamental cycle passing though $(u,v)$ has an odd number of negative edges. They then return up to 5 such cycles sorted by their length.

We experimented with two different modifications:
\begin{enumerate}
\item[{\tt FR}1:] return {\bf all} fundamental cycles that are frustrated (in a single spanning tree).
\item[{\tt FR}:~\!~] construct a (not necessarily spanning) BFS tree of depth at most $d_{\max}=3$ from each vertex of $\calG$, return all frustrated fundamental cycles in these trees. Such algorithm is able to find cycles of length $2d_{\max}+1=7$.

Note that both {\tt FR} and {\tt SAC} do a BFS-like search from  $(r,s)$ for every $r\in V,s\in\calY_r$.
In general, the complexity of {\tt FR} increases with $d_{\max}$, while in the case of {\tt SAC} this is not necessarily the case since the AC algorithm may stabilize at some point.
\end{enumerate}

Our overall implementation uses the same structure as in the previous section. The only difference is that in lines 3-5 of Algorithm~\ref{alg:SAC} instead of sets $\calA_r$ we use the cycles computed as above, triangulated into triplets; in line 5 they are sorted according to their size.

In our tests the 5-cycle version from~\citep{sontag_frustrated_cycles} was always worse than ${\tt FR}1$ and ${\tt FR}$ (usually substantially). For this reason we will report the results of only the latter two algorithms.

We conclude this section with the following observation; its proof is given in the appendix.
\begin{theorem}\label{th:FR}
If $\calG$ has a frustrated cycle then the SAC tightening (Algorithm~\ref{alg:SAC}) will find at least one node $r$ with $\calL_r=\calY_r$, and hence will output triplets that will yield a strict increase of the lower bound.\footnote{This theorem is stated for the unbounded-depth versions of {\tt SAC} and {\tt FR} tightenings. It also holds if the search depth bound for {\tt SAC} allows it to find cycles of the same length as {\tt FR}. However, in our implementation, the search depth is increased if we do not find an obstruction, so the unbounded-depth statement of this theorem is sufficient for our purposes.}
\end{theorem}
Note that the converse does not necessarily hold: it may happen that $\calG$ does not have frustrated cycles but  {\tt SAC} still returns some triplets. For example, {\tt SAC}
is able to find triplets that increase costs $\theta_r(a)$ of some labels $a\in\calX_r$ but not all of them. This could correspond to cutting off some fractional vertices of the local polytope without necessarily improving the lower bound. In contrast, the frustrated cycle procedure is only able to find  tightenings that strictly improve the bound.

\section{Experiments}
\label{sec_experiments}

We compared the lower bounds computed by the following algorithms.
\begin{itemize}
\item  {\tt SAC}, {\tt FR}1 and {\tt FR}, as described in previous sections. We built our implementation as an extension of the SRMP software (\url{https://pub.ista.ac.at/~vnk/software.html#SRMP}, GPL License); the code is attached as a supplementary material. We start the process with 100 iterations of SRMP and then repeat the following procedure: find triplets using the current reparameterization, add them to the relaxation and run 100 iterations of SRMP.
We stop the process if no new triplets were found at a given iteration, or if the time limit of 300 seconds was reached.
\item SRMP algorithm without any tightening. We stop when it converged, or when it runs for more than 300 seconds. We show only the last point as a green dot.
\item toulbar2 (\url{https://github.com/toulbar2/toulbar2}, version v1.2.1, MIT License). It was the winner of the UAI 2022 competition on the task of computing the best primal solution for MAP-MRF problems. We run three versions of toulbar2 - vns, ipr and vacint, and take the best. 

\end{itemize}

All codes were written in C++. We ran it on a laptop running Ubuntu 22.04.2 LTS (Linux kernel 5.19.0-50-generic), equipped with an Intel Core i5-1335U CPU (10 cores, 12 threads, up to 4.6 GHz), 12 MiB L3 cache, and 15 GB RAM.

For the comparison, we used the benchmark of examples from the UAI 2022 competition\footnote{https://uaicompetition.github.io/uci-2022/} which includes MAP-MRF instances (called ``MPE''). It has 17 classes of problems with unary and pairwise terms.
We selected the first instance from each
class, and excluded 3 instances on which the methods converged immediately to an optimal solution (without any need for tightening). The plots for the remaining 14 instances are shown below, and tables with numerical results are given in the appendix.

The title of each plot gives the name of the file
and the size of the instance in the form $(n, m, \#\ell)$ where $n$ is the number of nodes, $m$ is the number of edges, and $\ell$ is the number of labels which can be a range in the case of varying number of labels. If a trend is interrupted early, this means that the algorithm terminated, as no further progress could be made.

\centering
\includegraphics[width=0.45\textwidth]{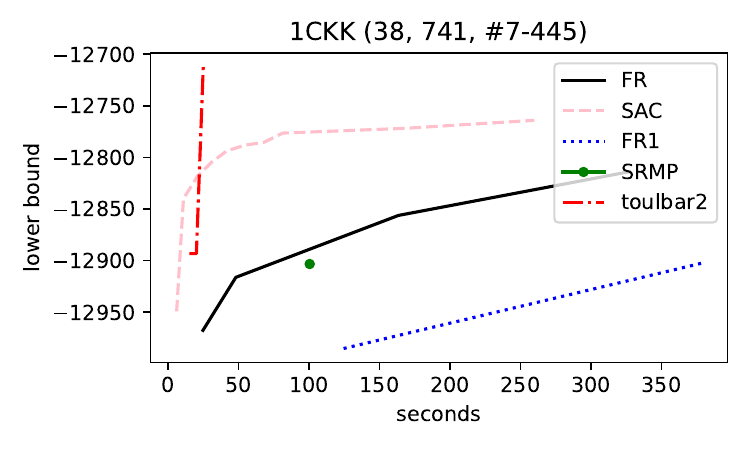}
\hspace{0.05\textwidth}
\includegraphics[width=0.45\textwidth]{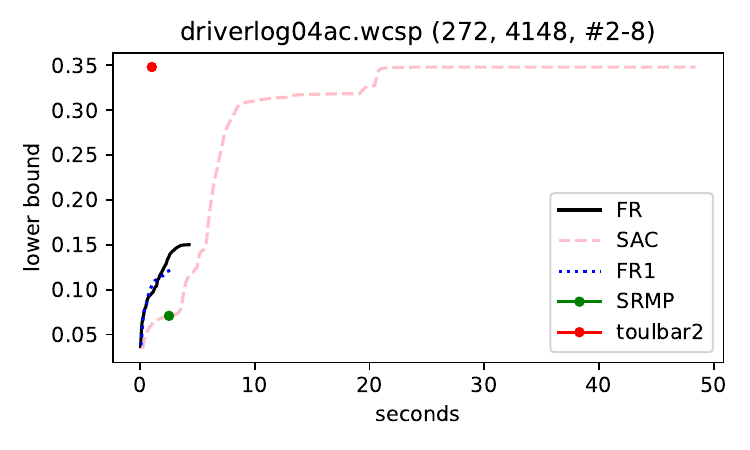}
\hspace{0.05\textwidth}
\includegraphics[width=0.45\textwidth]{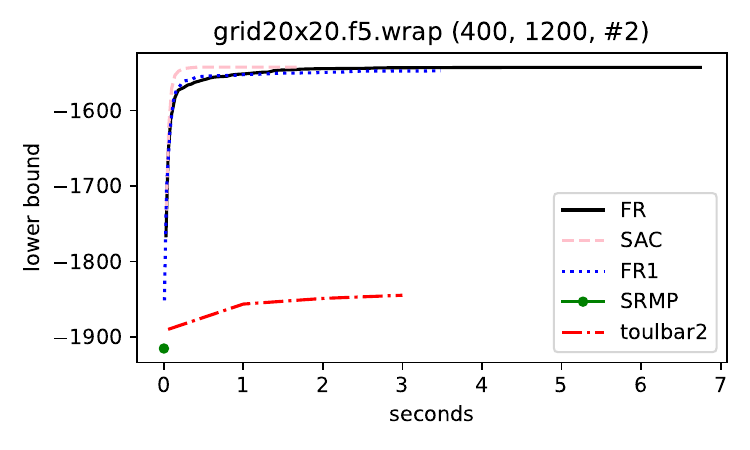}
\vspace{0.5cm}
\includegraphics[width=0.45\textwidth]{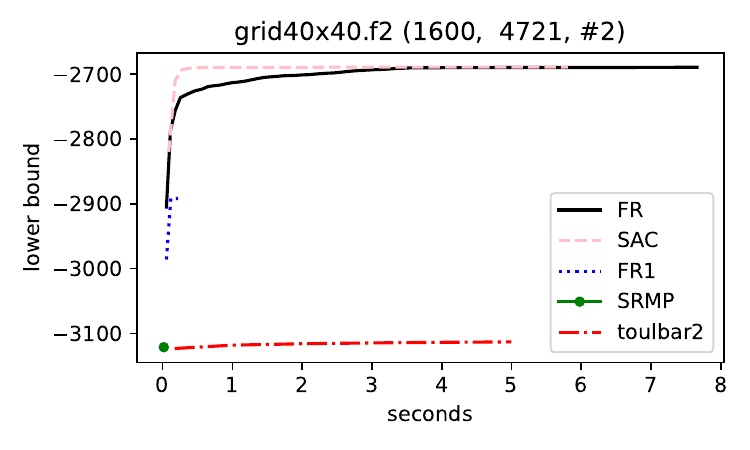}
\hspace{0.05\textwidth}
\includegraphics[width=0.45\textwidth]{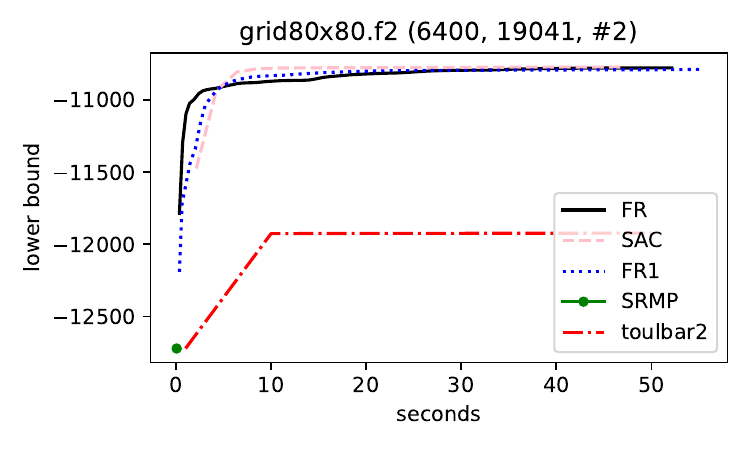}
\vspace{0.5cm}
\includegraphics[width=0.45\textwidth]{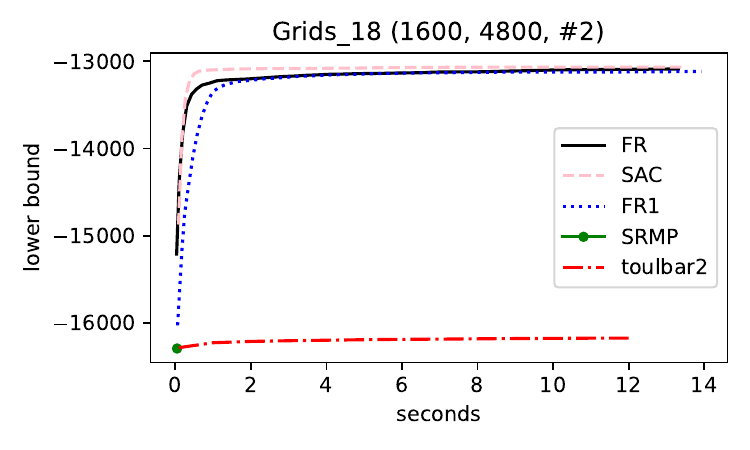}
\hspace{0.05\textwidth}
\includegraphics[width=0.45\textwidth]{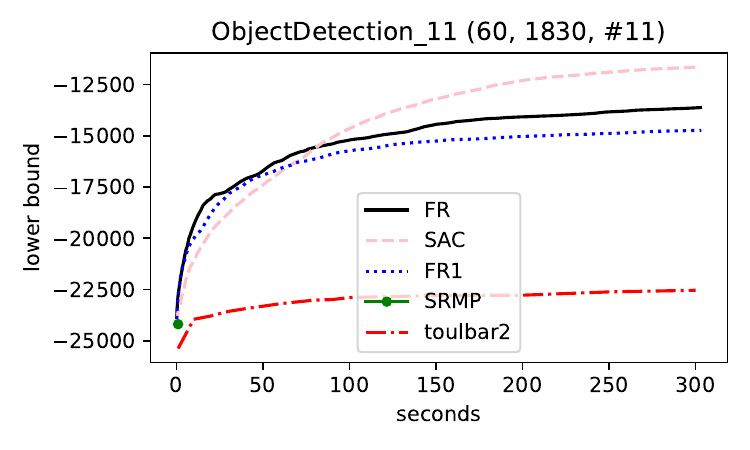}
\vspace{0.5cm}
\includegraphics[width=0.45\textwidth]{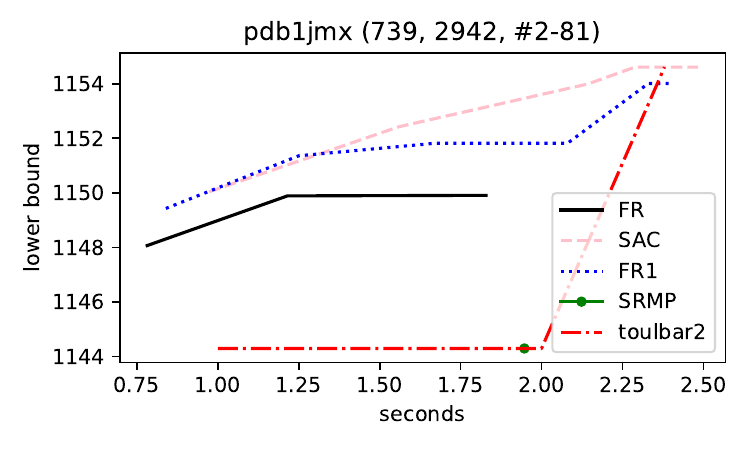}
\hspace{0.05\textwidth}
\includegraphics[width=0.45\textwidth]{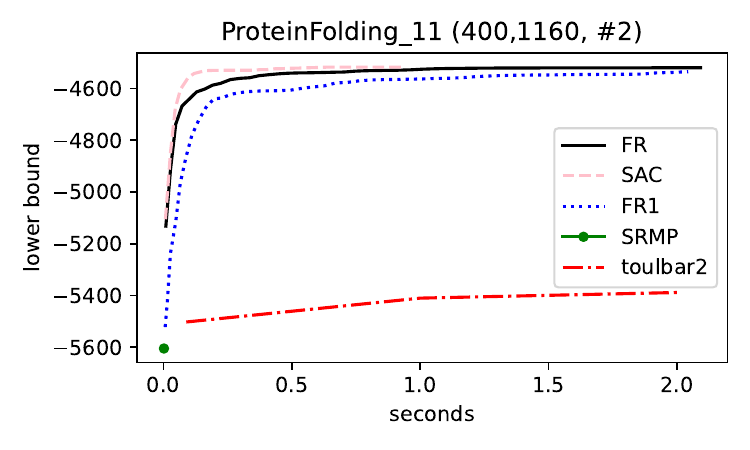}
\vspace{0.5cm}
\includegraphics[width=0.45\textwidth]{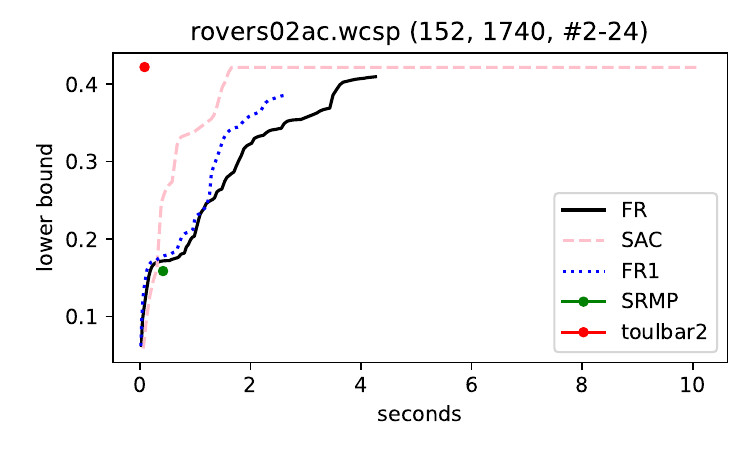}
\hspace{0.05\textwidth}
\includegraphics[width=0.45\textwidth]{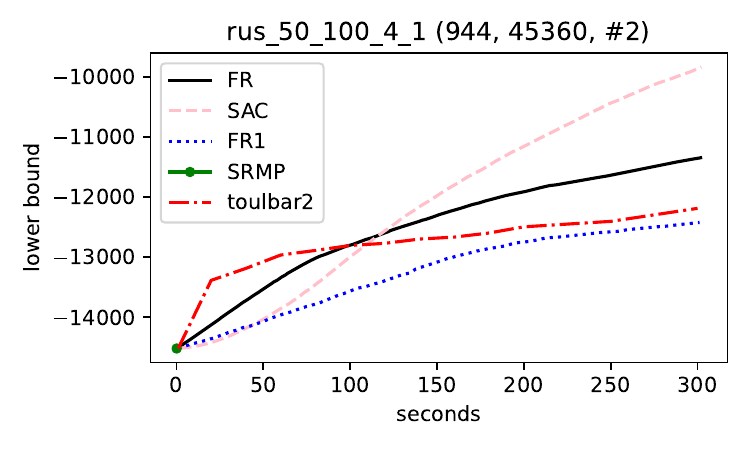}
\vspace{0.5cm}
\includegraphics[width=0.45\textwidth]{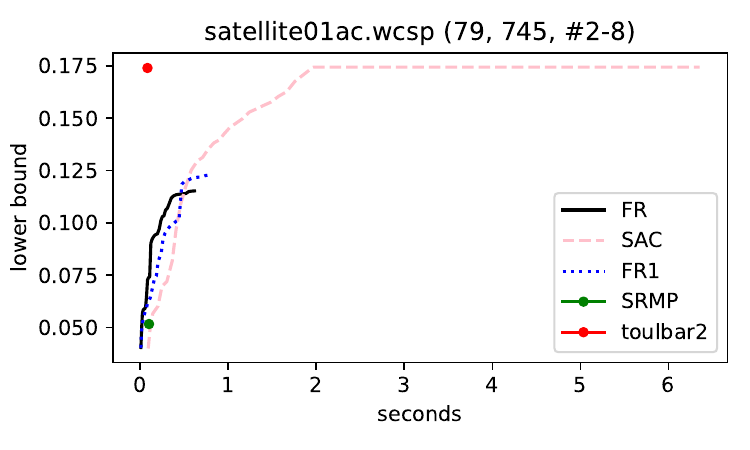}
\hspace{0.05\textwidth}
\includegraphics[width=0.45\textwidth]{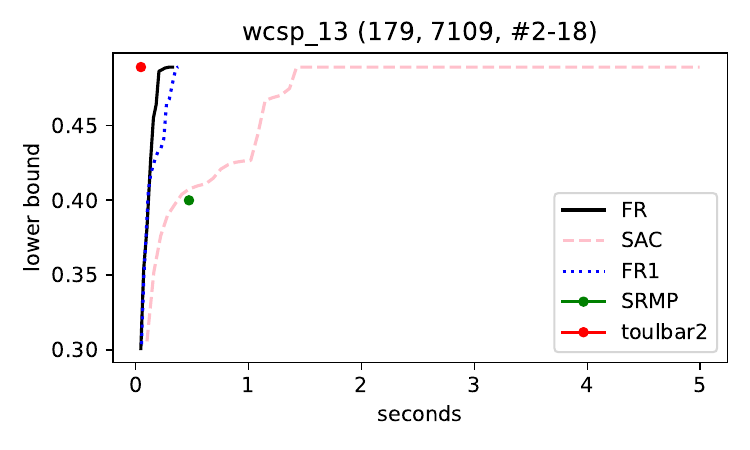}
\vspace{0.5cm}
\includegraphics[width=0.45\textwidth]{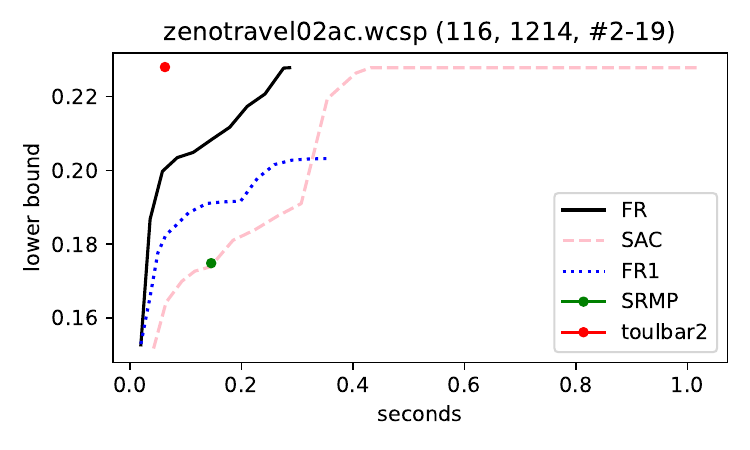}

\justifying
\myparagraph{Discussion of results} 
From the plots one can see that the highest lower bound attained by {\tt SAC} is the same or better than that of {\tt FR}1 / {\tt FR}, and in half of the cases it is strictly better.
This confirms the discussion in Section~\ref{sec:frustrated} that {\tt SAC} is capable of finding more useful triplets than {\tt FR}1 / {\tt FR}.

Let us now discuss the performance of {\tt SAC} vs toulbar2. The former finds
a much better bound in 7 cases, a worse bound in 1 case (1CKK), and similar bounds in 6 cases. In these 6 cases toulbar2 reaches the best bound faster than {\tt SAC} in all cases except for one.
We remark toulbar2 uses a very different strategy; according to the documentation, its lower bounds are based on a combination of various versions of Virtual Arc Consistency (VAC), dead end elimination and branch-and-bound search. As the plots show, in many cases this strategy is less powerful than the cluster-based tightening. 

Overall, we conclude that {\tt SAC} is the most robust tightening technique among considered approaches. As noted above, {\tt SAC} cannot make progress if the current CSP instance is singleton-arc-consistent, which is the limitation of {\tt SAC}. For the {\tt FR} technique, a similar limitation is an absence of frustrated cycles. Theorem~\ref{th:FR} shows that {\tt SAC}'s limitation is weaker than {\tt FR}'s limitation in the sense that we can make progress whenever {\tt FR} can. 

\section*{Acknowledgements}
We thank Ivan Sergeev for useful discussions.

\bibliographystyle{alpha}
\bibliography{references}

\appendix
\section{Proofs from Section~\ref{sec_theory}}

\begin{proof}[Proof of Theorem~\ref{th_strict_tightening}]
    Without loss of generality, assume that $\min_{x_\alpha} \theta_{x_\alpha}(x_\alpha) = 0$ for all $\alpha \in \calF'$. Now, to prove the theorem, it is enough to find a reparameterization satisfying $\min_{x_\alpha}\thetabar_\alpha(x_\alpha) \geq 0$ for all $\alpha\in\calF'$, and $\thetabar_r(s)>\theta_r(s)$. 

    Let us call a pair of labels $(a, b)\in \calX_u \times \calX_v$ \emph{active} if the corresponding cost $\theta_{uv}(a,b) \leq \varepsilon$ and \emph{inactive} otherwise. Thus, for every inactive pair $(a,b)\in \calX_u \times \calX_v$ we have $\theta_{uv}(a, b) > \varepsilon$. Equivalently, $(a, b)$ is active if and only if $(a, b) \in \calY_{uv}$ in the beginning of the AC run. The active/inactive status for a pair remains unchanged throughout the proof.

    We are going to follow the sequence of label deletions from the AC run and simultaneously edit the costs $\thetabar$ to construct a reparameterization. We start with $\thetabar = \theta$. We will show that after the deletion of $a_i$, $i < K$, the following properties hold.
    \begin{enumerate}
        \item For any factor $\alpha \in \calF'$, we have $\min_{x_\alpha}\thetabar_\alpha(x_\alpha) \geq 0$.
        \item For any deleted label $a_j \in v_j$, $j \leq i$, we have $\thetabar_{r v_j}(s, a_j) \geq 2^{-i+1} \varepsilon$.
        \item For any inactive pair $(a, b) \in \calX_u \times \calX_v$, we have $\thetabar_{uv}(a, b) \geq 2^{-i+1} \varepsilon$.
    \end{enumerate}

    We will use induction on $i$. Base case: consider the deletion of $a_1 \in \calX_{v_1}$. The deletion happened, since all active pairs $(b, a_1) \in \calX_{r} \times \calX_{v_1}$ have $b \neq s$. Thus, property (2) holds. Properties (1) and (3) hold trivially.

    Step case: consider the deletion of $a_i \in \calX_{v_i}$. If $u_i = r$, then proceed as in the base case. Otherwise, $u_i \neq r$. For all active pairs $(b, a_i) \in \calX_{u_i} \times \calX_{v_i}$, $b$ was deleted at some step in the past. Therefore, by property (2), we have $\thetabar_{r u_i}(s, b) \geq 2^{-i+2} \varepsilon$. For all inactive pairs $(c, a_i) \in \calX_{u_i} \times \calX_{v_i}$ we have $\thetabar_{r u_i}(s, c) \geq 2^{-i+2} \varepsilon$ by property (3). Perform the following changes:
    \begin{equation}
    \begin{aligned}
        &\begin{cases}
            \thetabar_{u_i, v_i}(c, a_i) \leftarrow \thetabar_{u_i, v_i}(c, a_i) - 2^{-i+1} \varepsilon \\
            \thetabar_{r u_i v_i}(*, c, a_i) \leftarrow \thetabar_{r u_i v_i}(*, c, a_i) + 2^{-i+1} \varepsilon
        \end{cases}
        \quad &&\text{for } (c, a_i) \text{ inactive} \\ 
        &\begin{cases}
            \thetabar_{r, u_i}(s, b) \leftarrow \thetabar_{r, u_i}(s, b) - 2^{-i+1} \varepsilon \\
            \thetabar_{r u_i v_i}(s, b, *) \leftarrow \thetabar_{r u_i v_i}(s, b, *) + 2^{-i+1} \varepsilon
        \end{cases}
        \quad &&\text{for } (b, a_i) \text{ active} \\ 
        &\begin{cases}
            \thetabar_{r, v_i}(s, a_i) \leftarrow \thetabar_{r, v_i}(s, a_i) + 2^{-i+1} \varepsilon \\
            \thetabar_{r u_i v_i}(s, *, a_i) \leftarrow \thetabar_{r u_i v_i}(s, *, a_i) - 2^{-i+1} \varepsilon
        \end{cases}
    \end{aligned}
    \label{eq_reparameterization_proof}
    \end{equation}
    Every bracket corresponds to a single reparameterization move. The symbol $*$ means ``for every label''.

    Now let us verify that the invariants still hold. 
    \begin{enumerate}
        \item Let us check the first property. We never change $\thetabar_v(a)$, so property (1) holds for $|\alpha| = 1$.

        Consider the case $|\alpha = 2|$. For the edge $(u_i, v_i)$, we have decreased the values of $\thetabar_{u_i, v_i}(c, a_i)$ by $2^{-i+1} \varepsilon$ for $(c, a_i)$ inactive. By property (3), these values used to be at least $2^{-i+2} \varepsilon$. Thus, $\min\limits_{(a, b) \in \calX_u \times \calX_v} \thetabar_{uv}(a, b) \geq 0$.

        We have also decreased $\thetabar_{r, u_i}(s, b)$ for $(b, a_i)$ active. Since $b$ was deleted before, we had $\thetabar_{r, u_i}(s, b) \geq 2^{-i+2} \varepsilon$ by property (2) before this step. Therefore, after the current step, we have $\thetabar_{r, u_i}(s, b) \geq 2^{-i+1} \varepsilon$, and $\min\limits_{(a, b) \in \calX_r \times \calX_{v_i}} \thetabar_{r v_i}(a, b) \geq 0$. Thus, property (1) also holds for $|\alpha| = 2$.

        Finally, consider the case $|\alpha = 3|$. Consider the values $\theta_{r u_i v_i}$. In the third bracket from Equation~\eqref{eq_reparameterization_proof}, we have decreased $\thetabar_{r u_i v_i}(s, *, a_i)$. It is enough to show that these values were increased (by at least the same amount) in the first two brackets. Indeed, for the triples $(s, c, a_i)$ with $(c, a_i)$ inactive, we have increased the term $\thetabar_{r u_i v_i}(s, c, a_i)$ in the first bracket. For the triples $(s, b, a_i)$ with $(b, a_i)$ active, we have increased the term $\thetabar_{r u_i v_i}(s, b, a_i)$ in the second bracket. Thus, none of the $\thetabar_{\alpha}(a_\alpha)$ were decreased for $|\alpha = 3|$, and property (1) holds.

        \item For any label $a_j$ that was deleted before this round (so, $j < i$), we had $\thetabar_{r v_j}(s, a_j) \geq 2^{-i+2} \varepsilon$. Potentially decreasing $\thetabar_{r v_j}(s, a_j)$ by $2^{-i+1} \varepsilon$ does not bring the value below $2^{-i+1} \varepsilon$, so property (2) holds for these $a_j$. For the currently deleted $a_i$, in the third bracket of Equation~\eqref{eq_reparameterization_proof}, we have increased $\thetabar_{r v_j}(s, a_j)$ by $2^{-i+1} \varepsilon$, so property (2) also holds.

        \item Before this step, for any inactive pair $(a, b) \in \calX_u \times \calX_v$, we had $\thetabar_{uv}(a, b) \geq 2^{-i+2} \varepsilon$. Some inactive pairs got decreased by $2^{-i+1} \varepsilon$ in the first bracket of Equation~\eqref{eq_reparameterization_proof}. Thus, we have the inequality $\thetabar_{uv}(a, b) \geq 2^{-i+1} \varepsilon$, and property (3) holds.
    \end{enumerate}

   Consider the final step. Recall that $a_K = s$. The deletion happened, since all the labels from $u_K$ were either deleted, or did not form an active pair with $s$. By properties (2) and (3), we have $\thetabar_{r u_K}(s, a) \geq 2^{-K + 2} \varepsilon$ for all $a \in \calX_{u_K}$. We can do the reparameterization
    \begin{equation}
        \begin{cases}
            \thetabar_{r u_K}(s, *) \leftarrow \thetabar_{r u_K}(s, *) - 2^{-K + 2} \varepsilon \\
            \thetabar_r(s) \leftarrow \thetabar_r(s) + 2^{-K + 2} \varepsilon
        \end{cases}
    \end{equation}
    After this, no minimum of any term has decreased, but $\thetabar_r(s)$ has a positive increase, as claimed.
\end{proof}

\begin{proof}[Proof of Theorem~\ref{th_tightening_group_branching}]
Consider the following reparameterization. As in the proof of Theorem~\ref{th_strict_tightening}, assume that $\min_{x_\alpha} \theta_{x_\alpha}(x_\alpha) = 0$ for all $\alpha \in \calF'$, and start with the reparameterization $\thetabar = \theta$. For every edge $(A, C) \in \calE'$, perform the following changes. Assume $A \subseteq \calX_u$, $C \subseteq \calX_v$.
    \begin{equation}
    \begin{aligned}
        &\begin{cases}
            \thetabar_{u, v}(a, c) \leftarrow \thetabar_{u, v}(a, c) - f_{C} \\
            \thetabar_{r u v}(*, a, c) \leftarrow \thetabar_{r u v}(*, a, c) + f_{C}
        \end{cases}
        \quad &&\text{for } (a, c) \text{ inactive, } a \in \calX_u, c \in C \\ 
        &\begin{cases}
            \thetabar_{r, u}(s, a) \leftarrow \thetabar_{r, u}(s, a) - f_{C} \\
            \thetabar_{r u v}(s, a, *) \leftarrow \thetabar_{r u v}(s, a, *) + f_{C}
        \end{cases}
        \quad &&\text{for } a \in A \\ 
        &\begin{cases}
            \thetabar_{r, v}(s, c) \leftarrow \thetabar_{r, v}(s, c) + f_{C} \\
            \thetabar_{r u v}(s, *, c) \leftarrow \thetabar_{r u v}(s, *, c) - f_{C}
        \end{cases}
        \quad &&\text{for } c \in C \\ 
    \end{aligned}
    \label{eq_reparameterization_group}
    \end{equation}
    where $f_C$ are some positive numbers to be determined.
    Finally, perform the last step. Observe that all the incoming edges to the sink node $S$ are directed from labels in $\calX_{v_{K-1}}$.
    \begin{equation}
        \begin{cases}
            \thetabar_{r v_{K-1}}(s, *) \leftarrow \thetabar_{r v_{K-1}}(s, *) - f_S \\
            \thetabar_r(s) \leftarrow \thetabar_r(s) + f_S
        \end{cases}
    \end{equation}

    We wish to maximize the value $f_S$ while maintaining $\min_{x_\alpha} \theta_{x_\alpha}(x_\alpha)$ non-decreased. To ensure this, it is sufficient if the values $f_C$ satisfy the following constraints:
    \begin{equation}
        \label{eq_constraints_reparameterization_group}
        \begin{aligned}
            f_C &\leq \varepsilon && \quad \forall C, \\
            f_C &\geq \sum_{A : (C,A) \in \calE'} f_A && \quad \forall C. \\
        \end{aligned}
    \end{equation}
    Indeed, the decrease of $\min_{x_\alpha} \theta_{x_\alpha}(x_\alpha)$ could potentially happen in either of the three brackets of Equation~\eqref{eq_reparameterization_group} or in the final step. 
    
    Consider the first bracket. The decrease of $\thetabar_{u, v}(a, c)$ for $(a, c)$ inactive, $a \in \calX_u$, $c \in C$, happens only once for fixed $a$ and $c$. Therefore, the constraint $f_C \leq \varepsilon$ ensures that $\thetabar_{u, v}(a, c) \geq 0$. 
    
    Consider the second bracket and the last-step-bracket. The decrease of $\thetabar_{r, u}(s, a)$ happens for every $a \in A$ and every outgoing edge of $\calG'$ from $A$. Notice that the value $\thetabar_{r, u}(s, a)$ was increased earlier, say, at a reparameterization step corresponding to $(D, A) \in \calE'$, and $\thetabar_{r, u}(s, a)$ was increased by $f_A$ in the third bracket. Then, this value stays positive if we have $f_A \geq \sum_{C : (A,C) \in \calE'} f_C$, which is precisely the second constraint, up to variable renaming. 
    
    Consider the third bracket. As in the proof of Theorem~\ref{th_strict_tightening}, the terms $\thetabar_{r u_i v_i}(s, *, a_i)$ actually do not decrease, because they are increased by the same amount in either the first or the second brackets. Indeed, consider any triple $(s, a, c)$ such that $a \in \calX_u$, $c \in C$. If $a \in A$ then $\thetabar_{r u v}(s, a, c)$ was increased in the second bracket. If $a \notin A$, then $(a, c)$ must be inactive, and $\thetabar_{r u v}(s, a, c)$ was increased in the first bracket.

    We observe that for a fixed value of $f_S$, we have $f_A \geq B(A) f_S$. This follows from a bottom-up induction on the vertices of $\calG'$. Moreover, the equality $f_A = B(A) f_S$ can be attained if in Equation~\eqref{eq_constraints_reparameterization_group}, the second inequality is always met with equality. 
    
    Therefore, to meet the requirements $f_A \leq \varepsilon$, we need $f_S \leq \frac{\varepsilon}{\max\limits_A B(A)}$. Moreover, in case $f_S = \frac{\varepsilon}{\max\limits_A B(A)}$, there exist values $\{f_A\}$ that satisfy the constraints from Equation~\eqref{eq_constraints_reparameterization_group}.
\end{proof}

\begin{proof}[Proof of Corollary~\ref{cor_cycle}]
    In the setting of this corollary, the graph $\calG$ winds around the cycle $w_1, \ldots, w_d$. We will define a suitable construction of the graph $\calG'$ and then apply Theorem~\ref{th_tightening_group_branching}. 
    
    Let us construct $\calG'$ as follows: for each vertex $w_i$, contract all the nodes $a$ of $\calG$ that satisfy $a \in \calX_{w_i}$. This is a valid construction of the graph $\calG'$, since we never introduce a cycle, and all the edges that are incoming to $a \in \calX_{w_i}$ are coming from the same label set $\calX_{w_{i-1}}$.

    The constructed $\calG'$ is a path. Therefore, $B(A) = 1$ for all $A \in \calV'$. By Theorem~\ref{th_tightening_group_branching}, there exists the desired reparameterization that increases $\thetabar_r(s)$ by $\varepsilon$.
\end{proof}

\section{Proof of Theorem \ref{th:FR}}\label{sec:th:FR:proof}
Define $\calI=CSP_{\varepsilon}^{\le 2}(\theta)$.
Let $\calC=\{\pi_1,\ldots,\pi_k\}$ be a frustrated cycle in $\calG$, with $\pi_i\in\Pi_{v_i}$. Denote $r=v_1$,
and consider label $s\in\calY_r$. We have $s\in\pi_1^{\sigma_1}$
where either $\sigma_1=+$ or $\sigma_1=-$.
Let us fix the label of $r$ to $s$,
i.e.\ update $\calY_r\leftarrow\{s\}$. It suffices to show that running AC on $\calY$ will remove $s$ from $\calY_r$.
For $i\in[k+1]$ define $\sigma_i=\sigma_1$ if the path $(\pi_1,\ldots,\pi_i)$ contains an even number of negative edges, and $\sigma_i=-\sigma_1$
otherwise (where we assume that $\pi_{k+1}=\pi_1$).
We will prove by induction on $i$
that after a certain number of AC steps we will have $\calY_{v_i}\subseteq \pi_i^{\sigma_i}$.
For $i=0$ the claim holds since in the beginning $\calY_{v_1}=\{s\}\subseteq\pi_1^{\sigma_1}$.
Suppose that  it holds for $i\in[k]$; let us prove it for $i+1$. We will only consider the case when edge  
$\{\pi_i,\pi_{i+1}\}$ is positive in $\calG$ (and hence $\sigma_{i+1}=\sigma_i=\sigma$);
for negative edges the argument is analogous.
Since $\{\pi_i,\pi_{i+1}\}$ is positive, for all $(a,b)\in\pi_i^\sigma\times\pi_{i+1}^{-\sigma}$ we must have $\theta_{v_iv_{i+1}}(a,b)> \min_{(a',b')}\theta_{v_iv_{i+1}}(a',b')+\varepsilon$ and so $(a,b)\notin\calY_{v_iv_{i+1}}$.
Also, $\calY_{v_i}\subseteq\pi^\sigma_i$ by the induction hypothesis.
Therefore, the AC algorithm will remove all labels from $\pi^{-\sigma}_{i+1}$. The claim follows.

Since $\calC$ is a frustrated cycle, we must have $\sigma_{k+1}=-\sigma_1$, and so $\calY_r\subseteq\pi_{k+1}^{-\sigma_1}=\pi_1^{-\sigma_1}\subseteq\calX_v-\{s\}$ after some number of AC steps. We also have $\calY_r\subseteq\{s\}$, and hence $\calY_r=\varnothing$.

We saw that fixing every label $s \in \calY_r$ leads to an AC contradiction. This implies that adding all the triplets generated by SAC tightening will admit a reparametrization that strictly increases the lower bound. Indeed, for every label $s \in \calY_r$, we can run the reparametrization constructed in the proof of Theorem~\ref{th_tightening_group_branching}, that increases the value $\thetabar_r(s)$. Let $\lambda^s$ be the corresponding reparametrization vector. Taking any reparametrization from the interior of the convex hull of $\{\lambda^s\}_{s \in \calY_r}$ yields a strict increase of the lower bound. The costs $\thetabar_r(s)$ are increased for every $s \in \calY_r$, and the constraint $\min_{x_\alpha}\thetabar_\alpha(x_\alpha)\ge \min_{x_\alpha}\theta_\alpha(x_\alpha)$ holds for every term $\alpha$, since these constraints are linear in $\lambda$ and hold for every reparametrization vector $\lambda^s$.

\section{Tables}

In this paragraph we give tables for the results. In case runnings were stopped at 300 seconds, we give the best lower bounds obtained in 50, 100, 200 and 300 seconds. Otherwise, we report the best lower bound that was obtained and the corresponding runtime.

\subsection{1CKK}
\begin{tabular}{ |c| c|c|c|c|} 
\hline
algorithm & 50 s & 100 s & 200 s & 300 s \\
\hline
FR &  -12916.2& -12916.2 & -12856.2 & -12856.2 \\
\hline
SAC& -12793.9& -12776.4& -12771.8 & -12764 \\
\hline
FR1 &&-12985.2 & -12985.2 & -12985.2 \\
\hline
toulbar2&  -12712.414 &&&\\
\hline
\end{tabular}

Here, toulbar2 finds the best bound (-12712.414) in 24.84 seconds.

\subsection{driver}
\begin{tabular}{|c|c|c|}
\hline
algorithm & best bound & time [s] \\
\hline
FR & 0.15009 & 4.28 \\
\hline
SAC & 0.347921 & 48.43 \\
\hline
FR1 & 0.122707 & 2.82 \\
\hline
toulbar2 & 0.348 & 1.03 \\
\hline
\end{tabular}

\subsection{grid20x20.f5.wrap}
\begin{tabular}{|c|c|c|}
\hline
algorithm & best bound & time [s] \\
\hline
FR & -1543 & 6.74 \\
\hline
SAC & -1542.66 & 48.43 \\
\hline
FR1 & -1547.58 & 3.48 \\
\hline
toulbar2 & -1844.635 & 3.00 \\
\hline
\end{tabular}

\subsection{grid40x40.f2}
\begin{tabular}{|c|c|c|}
\hline
algorithm & best bound & time [s] \\
\hline
FR & -2689.35 & 7.66 \\
\hline
SAC & -2689 & 5.81 \\
\hline
FR1 & -2891.79 & 0.22 \\
\hline
toulbar2 & -3112.997 & 5.00 \\
\hline
\end{tabular}

\subsection{grid80x80.f2}
\begin{tabular}{|c|c|c|}
\hline
algorithm & best bound & time [s] \\
\hline
FR & -10778.2 & 52.14 \\
\hline
SAC & -10772.9 & 47.05 \\
\hline
FR1 & -10789.2 & 55.25 \\
\hline
toulbar2 & -11922.615 & 50.00 \\
\hline
\end{tabular}

\subsection{Grids\_18}
\begin{tabular}{|c|c|c|}
\hline
algorithm & best bound & time [s] \\
\hline
FR & -13088 & 13.31 \\
\hline
SAC & -13065.9 & 13.51 \\
\hline
FR1 & -13116.7 & 13.92 \\
\hline
toulbar2 & -16172.240 & 12.00 \\
\hline
\end{tabular}

\subsection{ObjectDetection\_11}
\begin{tabular}{ |c| c|c|c|c|} 
\hline
algorithm & 50 s & 100 s & 200 s & 300 s \\
\hline
FR &  -16697.8& -15185.9 & -14073.4 & -13627.3 \\
\hline
SAC& -172463.3& -14584.1& -12286.3 & -11653 \\
\hline
FR1 & -16877.4& -15772.3 & -15035.2 & -14739.2 \\
\hline
toulbar2&  -23310.33 &-22888.296&-22773.707&-22533.442\\
\hline
\end{tabular}

\subsection{pdb1jmx}
\begin{tabular}{|c|c|c|}
\hline
algorithm & best bound & time [s] \\
\hline
FR & 1149.91 & 1.83 \\
\hline
SAC & 1154.61 & 2.48 \\
\hline
FR1 & 1154.01 & 2.39 \\
\hline
toulbar2 & 1154.608 & 2.38 \\
\hline
\end{tabular}

\subsection{ProteinFolding\_11}
\begin{tabular}{|c|c|c|}
\hline
algorithm & best bound & time [s] \\
\hline
FR & -4519.97 & 2.09 \\
\hline
SAC & -4517.32 & 0.93 \\
\hline
FR1 & -4535.18 & 2.04 \\
\hline
toulbar2 & -5388.646 & 2.00 \\
\hline
\end{tabular}

\subsection{rovers02ac.wcsp}
\begin{tabular}{|c|c|c|}
\hline
algorithm & best bound & time [s] \\
\hline
FR & 0.409512 & 4.26 \\
\hline
SAC & 0.421597 & 10.12 \\
\hline
FR1 & 0.38553 & 2.65 \\
\hline
toulbar2 & 0.422 & 0.09 \\
\hline
\end{tabular}

\subsection{rus\_50\_100\_4\_1}
\begin{tabular}{ |c| c|c|c|c|} 
\hline
algorithm & 50 s & 100 s & 200 s & 300 s \\
\hline
FR &  -13525 & -12791.8 & -11908.8 & -11348.6 \\
\hline
SAC& -14030.7 & -12985.5 & -11140 & -9840 \\
\hline
FR1 & -14061.1 & -13563.1 & -12749 & -12424.1 \\
\hline
toulbar2&  -12966.472 & -12807.1 & -12497.016 & -12189.966 \\
\hline
\end{tabular}

\subsection{satellite01ac.wcsp}
\begin{tabular}{|c|c|c|}
\hline
algorithm & best bound & time [s] \\
\hline
FR & 0.115266 & 0.62 \\
\hline
SAC & 0.174434 & 10.12 \\
\hline
FR1 & 0.122888 & 0.78 \\
\hline
toulbar2 & 0.174 & 0.09 \\
\hline
\end{tabular}

\subsection{wcsp\_13}
\begin{tabular}{|c|c|c|}
\hline
algorithm & best bound & time [s] \\
\hline
FR & 0.489001 & 0.33 \\
\hline
SAC & 0.489001 & 4.99 \\
\hline
FR1 & 0.489001 & 0.38 \\
\hline
toulbar2 & 0.489 & 0.05 \\
\hline
\end{tabular}

\subsection{zenotravel02ac.wcsp}
\begin{tabular}{|c|c|c|}
\hline
algorithm & best bound & time [s] \\
\hline
FR & 0.227827 & 0.29 \\
\hline
SAC & 0.227827 & 1.02 \\
\hline
FR1 & 0.203196 & 0.35 \\
\hline
toulbar2 & 0.228 & 0.06 \\
\hline
\end{tabular}


\end{document}